\documentclass[12pt]{article}
\usepackage{amsmath,amssymb,amscd,amsthm,color,graphics,verbatim}
\usepackage[margin=1.0in]{geometry}

\usepackage{graphicx}

\usepackage[enableskew]{youngtab}
\def\qed{\hfill $\vrule height 2.5mm  width 2.5mm depth 0mm $}

\newcommand{\Arctan}{\operatorname{Arctan}}

\newtheorem{theorem}{Theorem}
\newtheorem{proposition}[theorem]{Proposition}

\newtheorem{conjecture}[theorem]{Conjecture}

\theoremstyle{definition}

\newtheorem{remark}[theorem]{Remark}
\newtheorem{example}[theorem]{Example}

\begin{document}
$\,$\vspace{0mm}

\begin{center}
{\sf\huge Bethe's Quantum Numbers And}
\vspace{3mm}\\
{\sf\Huge Rigged Configurations}
\vspace{15mm}\\
{\textsf{\LARGE  ${}^{\mbox{\small a,b}}$Anatol N. Kirillov and
${}^{\mbox{\small c}}$Reiho Sakamoto}}
\vspace{5mm}\\
{\textsf {${}^{\mbox{\small{a}}}$Research Institute for Mathematical Sciences,}}
\vspace{-1mm}\\
{\textsf {Kyoto University, Sakyo-ku,}}
\vspace{-1mm}\\
{\textsf {Kyoto, 606-8502, Japan}}
\vspace{-1mm}\\
{\textsf {kirillov@kurims.kyoto-u.ac.jp}}
\vspace{4mm}\\
{\textsf {${}^{\mbox{\small{b}}}$The Kavli Institute for the Physics and}}
\vspace{-1mm}\\
{\textsf {Mathematics of the Universe (IPMU),}}
\vspace{-1mm}\\
{\textsf {The University of Tokyo,}}
\vspace{-1mm}\\
{\textsf {Kashiwa, Chiba, 277-8583, Japan}}
\vspace{4mm}\\
{\textsf {${}^{\mbox{\small{c}}}$Department of Physics,}}
\vspace{-1mm}\\
{\textsf {Tokyo University of Science, Kagurazaka,}}
\vspace{-1mm}\\
{\textsf {Shinjuku, Tokyo, 162-8601, Japan}}
\vspace{-1mm}\\
{\textsf {reiho@rs.tus.ac.jp}}
\vspace{-1mm}\\
{\textsf {reihosan@08.alumni.u-tokyo.ac.jp}}
\end{center}
\vspace{15mm}

\begin{abstract}
\noindent
We propose a method to determine the quantum numbers, which we call the rigged configurations,
for the solutions to the Bethe ansatz equations for the spin-1/2 isotropic Heisenberg model
under the periodic boundary condition.
Our method is based on the observation that the sums of Bethe's quantum numbers
within each string behave particularly nicely.
We confirm our procedure for all solutions for length 12 chain (totally 923 solutions).\medskip

\noindent
Keywords: Heisenberg model, Bethe ansatz, rigged configurations.
\end{abstract}

\pagebreak

\section{Introduction}
Bethe's seminal solution to the isotropic Heisenberg model
under the periodic boundary condition \cite{Bethe} in 1931
is one of the prototypical theories on the quantum integrable systems.
Basic procedure of the algebraic Bethe ansatz is as follows
(see, e.g., \cite{Faddeev,KorepinBook}).
First, we solve the set of algebraic equations called the Bethe ansatz equations
(see equation (\ref{eq:Bethe_ansatz}) in the main text).
Next, by using the solutions to the Bethe ansatz equations,
we construct eigenstates of the Hamiltonian (see equation (\ref{eq:Bethe_vector})).
The main problem which we will consider in the present paper is to show that
the whole procedure involved in the Bethe ansatz method is mathematically well-defined.
We pursue the problem in the same setting as was treated by Bethe.

For a long time it has been observed that there are several subtle points about the procedure.
One of such obstructions is the problem called the singular solutions
(see equation (\ref{eq:general_singular_solution})).
Recently we have much progress on this issue (see
\cite{AV,Beisert,NepomechieWang2013,HNS1,KiSa14,KirillovSakamoto2014b})
and now we attain fairly good understanding of the problem.
However, the very structure of the solutions to the Bethe ansatz equations itself
still remains in rather foggy situation.

In the analysis of the solutions to the Bethe ansatz equations,
it has been customary to assume that roots of the solutions take a particular form called the strings.
This idea was already apparent in the original Bethe's paper.
However, as we will explain shortly after, the notion ``string" has rather elusive nature
and for a long time it has been a difficult problem to understand how to utilize the string structure
in a proper manner.
One of the well-known attempts on this problem is Takahashi's theory \cite{Takahashi}.
Let us explain in more detail.
Let $N$ be the length of the spin chain and let $\ell$ be the number of down spins.
Then a main assumption in the derivation of Takahashi's quantum number is to suppose that the strings
take exactly the following form:\footnote{We want to point out that the assumption of the validity
of the string shape (\ref{eq:string}) allows to guess some explicit expressions for $q$-weight
multiplicities (Kostka--Foulkas polynomials) appearing in the representation theory of Lie algebras
of type $A$. All these formulas have been proven rigorously in \cite{KR} and lead to applications in
combinatorics, representation theory and discrete integrable systems.}
\begin{align}
\label{eq:string}
a+bi,\,a+(b-1)i,\,a+(b-2)i,\,\ldots,\,a-bi,\qquad
(a\in\mathbb{R},\,b\in\mathbb{Z}_{\geq 0}/2).
\end{align}

However, when $\ell\geq 3$, the assumption (\ref{eq:string})
has serious difficulty.
Indeed, as we can easily see in examples, the real part $a$ as well as the intervals
between successive roots are not unique if the lengths of the strings are larger than 2.
As the result, Takahashi's quantum numbers are not uniquely defined and also
they are not half-integers.
In a previous paper \cite{Sakamoto2015}, we proposed to seek an alternative to
Takahashi's quantum numbers.
In particular, in that paper we showed that the correct quantum number
for the exceptional real solution \cite{EKS:1992} is different from the one
derived by Takahashi's quantum numbers.

In the present paper, as a continuation of a previous work \cite{Sakamoto2015},
we propose a method to assign quantum numbers to strings of roots.
For this purpose, we start from the so-called Bethe's quantum numbers (see Section \ref{sec:bethe}).
Here, roughly speaking, Bethe's quantum numbers appear as phase factors
after taking logarithm of the Bethe ansatz equations.
By definition, Bethe's quantum numbers are uniquely defined and exactly integers or half-integers.
Then our basic observation is that the {\it sum of Bethe's quantum numbers} associated with
all roots of a given string behaves in a particularly simple manner.
This is a remarkable property since individual Bethe's quantum numbers
behave in a rather complicated way.
Based on this observation we propose a method to determine complete set of
the quantum numbers, which we call the rigged configurations,
from Bethe's quantum numbers.
We confirmed these observations for all solutions of $N=12$ case by using
the numerical data given in \cite{GD2,HNS1}.

As we noted previously, we do not have thorough understanding of the string
pattern which appear in the solutions to the Bethe ansatz equations.
However, our main aim in the present paper is to propose a method
which could be made mathematically rigorous once we obtain sound understanding
of the string structure.
Our expectation relies on the fact that Bethe's numbers are uniquely defined
and exactly half-integers.

To our best knowledge, Bethe's quantum numbers had been introduced and studied
in the original paper by Bethe \cite{Bethe}, formulas (37a) and (37b),
and has been studied more recently in \cite{LinR} in their study of the set of
real solutions to the Bethe ansatz equations for $\ell=2$, and in \cite{HC}
concerning the deviation of solutions to the Bethe ansatz equations from having
exact string structure with numerical evidence for $N=10$.

The outline of the present paper is as follows.
In Section \ref{sec:ABA}, we provide necessary facts about the Bethe ansatz method
and the rigged configurations.
In Section \ref{sec:bethe} we collect necessary facts about Bethe's quantum numbers.
In Section \ref{sec:main} we provide the main algorithm.
We provide examples in Section \ref{sec:examples}.
Finally we conclude in Section \ref{sec:conclusion}.
Besides the main article, there are supplementary tables (see remarks after
Conjecture \ref{conj:main}).\footnote{
After completing the present work, we noticed that the paper \cite{DG2} appeared.
In this paper the authors also study Bethe's quantum numbers.
However our work is significantly different from theirs since our main motivation
is to understand the situation which appear when $\ell\geq 3$
whereas their paper considers the case $\ell=2$ exclusively.
Note that by a result of \cite{Vla} the real part ``$a$" of the string (\ref{eq:string})
is uniquely defined when $\ell=2$ (except for the exceptional real solutions).}

\section{Algebraic Bethe ansatz}
\label{sec:ABA}
In this section we briefly overview some basic facts concerning the spin-1/2
isotropic Heisenberg model, also known as the XXX model.
The space of states of the spin-1/2 XXX chain of length $N$ is the complex space $\mathbb{C}^{2^N}$,
which one can identify with the tensor product of $N$ copies of $\mathbb{C}^2$, namely
\begin{align}
\mathfrak{H}_N &= \bigotimes_{j=1}^{N} V_j,\quad V_j \simeq {\mathbb {C}}^2.
\end{align}
The Hamiltonian describing the interaction of particles has the form
\begin{align}
\mathcal{H}_N &=  \frac{J}{4} \sum_{k=1}^{N}( \sigma_{k}^{x} \sigma_{k+1}^{x}
+ \sigma_{k}^{y} \sigma_{k+1}^{y}+ \sigma_{k}^{z} \sigma_{k+1}^{z} -{\mathbb {I}}_N),
\end{align}
where we assume $\sigma_{N+1}^{a}=\sigma_{1}^{a}$, $a\in\{x,y,z\}$, and
\begin{align}
\sigma^{x} = 
\left(\!
\begin{array}{cc}
0&1\\
1&0
\end{array}
\!\right),\qquad
\sigma^{y} = 
\left(\!
\begin{array}{cc}
0&-i\\
i&0
\end{array}
\!\right),\qquad
\sigma^{z} = 
\left(\!
\begin{array}{cc}
1&0\\
0&-1
\end{array}
\!\right)
\end{align}stand for the Pauli matrices, and
\begin{align}
\sigma_{k}^{a}= I \otimes \cdots \otimes
\underbrace{\sigma^{a}}_{k}
\otimes \cdots \otimes I.
\end{align}
Let us introduce the vectors
\begin{align}
v_+=
\left(\!
\begin{array}{c}
1\\0
\end{array}
\!\right)
\text{ and }
v_-=
\left(\!
\begin{array}{c}
0\\1
\end{array}
\!\right)
\end{align}
and the global vacuum (ground state) as
\begin{align}
|0\rangle_N=v_+\otimes\cdots\otimes v_+
\in\mathfrak{H}_N.
\end{align}

Let us consider the local transfer matrix
\begin{align}
\label{eq:def_L_of_Bethe}
L_{k}(\lambda)=
\left(\!
\begin{array}{cc}
\lambda \mathbb{I}_N+\frac{i}{2}\sigma^z_k & \frac{i}{2}\sigma^-_k\\
\frac{i}{2}\sigma^+_k & \lambda \mathbb{I}_N-\frac{i}{2}\sigma^z_k
\end{array}
\!\right)
\end{align}
where $\sigma^\pm_k =\sigma^x_k\pm i\sigma^y_k$, and consider the monodromy matrix
\begin{align}
T_N(\lambda)=L_N(\lambda)L_{N-1}(\lambda)\cdots L_1(\lambda)=
\left(
\begin{array}{cc}
A_N(\lambda) & B_N(\lambda)\\
C_N(\lambda) & D_N(\lambda)
\end{array}
\right).
\end{align}
Recall that the local transfer matrices satisfy  the relation
$R_{12}L_1L_2=L_2L_1R_{12}$ for some matrix $R\in\mathbb{C}^2\otimes\mathbb{C}^2$
where $R$ satisfies the quantum Yang--Baxter equation
\begin{align}
\label{eq:YBE}
R_{12}R_{13}R_{23}=R_{23}R_{13}R_{12}.
\end{align}
The quantum Yang--Baxter equation (\ref{eq:YBE}) implies that the transfer matrices
$\operatorname{tr}|_{\mathbb{C}_0^2}T_N(\lambda)$
and $\operatorname{tr}|_{\mathbb{C}_0^2}T_N(\mu)$ ($\mathbb{C}_0^2$ is the auxiliary space)
commute for different values of $\lambda$ and $\mu$.
As a consequence of the relation (\ref{eq:YBE}), we can deduce that $[B_N(\lambda),B_N(\mu)]=0$
for different values of the parameters $\lambda$ and $\mu$.
The key observation due to Bethe, reformulated in the language of the algebraic Bethe ansatz, is that the vector
\begin{align}
\label{eq:Bethe_vector}
\Psi_N(\lambda_1,\ldots,\lambda_\ell)
=B_N(\lambda_1)\cdots B_N(\lambda_\ell)|0\rangle_N
\end{align}
is an eigenstates of the Hamiltonian $\mathcal{H}_N$ if and only if the parameters $\lambda_1,\ldots,\lambda_\ell$
satisfy the following system of algebraic equations, known as the Bethe ansatz equations
($\operatorname{BAE}(\ell)$ for short)
\begin{align}
\label{eq:Bethe_ansatz}
\left(
\frac{\lambda_k+\frac{i}{2}}{\lambda_k-\frac{i}{2}}
\right)^N
=\prod_{j=1 \atop j\neq k}^\ell
\frac{\lambda_k-\lambda_j+i}{\lambda_k-\lambda_j-i},
\qquad
(k=1,\cdots,\ell).
\end{align}

We are interested in the so-called physical solutions to $\operatorname{BAE}(\ell)$.
The set of physical solutions is divided into the regular solutions, that is,
the solutions $(\lambda_1,\ldots,\lambda_\ell)$ which do not contain a pair
$(\lambda_\alpha,\lambda_\beta)=(\frac{i}{2},-\frac{i}{2})$,
and the set of physical singular solutions of the form
\begin{align}
\label{eq:general_singular_solution}
\lambda=
\left\{
\frac{i}{2},-\frac{i}{2},\lambda_3,\ldots,\lambda_\ell
\right\}.
\end{align}
In the case of the regular solutions, one has
\begin{align}
\mathcal{H}_N\Psi_N(\lambda_1,\ldots,\lambda_\ell)=
\mathcal{E}_{\lambda_1,\ldots,\lambda_\ell}
\Psi_N(\lambda_1,\ldots,\lambda_\ell),\qquad
\mathcal{E}_{\lambda_1,\ldots,\lambda_\ell}:=-\frac{J}{2}\sum_{j=1}^\ell\frac{1}{\lambda_j^2+\frac{1}{4}}.
\end{align}
However for singular solutions, the energy $\mathcal{E}_{\lambda_1,\ldots,\lambda_\ell}$ is not defined.
It was suggested to ``resolve" singularity of $\mathcal{E}_{\lambda_1,\ldots,\lambda_\ell}$
by using a deformed form of singular solutions, namely,
\begin{align}
\label{eq:regularization_of_Nepomechie}
\widetilde{\lambda}_1=\frac{i}{2}+\epsilon+c\,\epsilon^N,\qquad
\widetilde{\lambda}_2=-\frac{i}{2}+\epsilon,
\end{align}
and consider the limit
\begin{align}
\lim_{\epsilon\rightarrow 0}\frac{1}{\epsilon^N}
B_N(\widetilde{\lambda}_1)B_N(\widetilde{\lambda}_2)
B_N(\lambda_3)\cdots B_N(\lambda_\ell)
=\widetilde{\Psi}_\lambda.
\end{align}
It was conjecture by \cite{NepomechieWang2013}, that $\widetilde{\Psi}_\lambda\neq 0$
if and only if
\begin{align}
\label{eq:NWcondition}
\left(
-\prod_{j=3}^\ell
\frac{\lambda_j+\frac{i}{2}}{\lambda_j-\frac{i}{2}}
\right)^N=1,
\end{align}
and
\begin{align}
c=2i^{N+1}\prod_{j=3}^\ell
\frac{\lambda_j+\frac{3i}{2}}{\lambda_j-\frac{i}{2}}.
\end{align}
We call the singular solutions which satisfy the condition (\ref{eq:NWcondition}) the physical singular solution.
The corresponding energy is \cite{KirillovSakamoto2014b}
\[
\mathcal{E}_\lambda=-J-\frac{J}{2}\sum_{j=3}^\ell\frac{1}{\lambda_j^2+\frac{1}{4}}.
\]
It is also conjectured \cite{HNS1} that the set of physical singular and regular solutions to
$\operatorname{BAE}(\ell)$ exhausts the set of all physical solutions to $\operatorname{BAE}(\ell)$.

Finally, let us define the rigged configurations.
The rigged configurations $(\nu,I)$ are comprised of a partition
$\nu=(\nu_1,\ldots,\nu_l)$, called the configuration, and a sequence of
non-negative integers $I_j$, called the riggings, satisfying the following condition.
Define the vacancy number $P_{k}(\nu)$ by the following formula
\begin{align}
P_{k}(\nu)=N-2(\nu'_1+\cdots +\nu'_k)
=N-2\sum_{j=1}^l\min(k,\nu_j)
\end{align}
where $\nu'_j$ represent components of the transposed partition $\nu'$.
If we regard the rigged configuration as the set of pairs
$(\nu,I)=\{(\nu_1,I_1),\ldots,(\nu_l,I_l)\}$, then the riggings must satisfy
\begin{align}
0\leq I_j\leq P_{\nu_j}(\nu)
\end{align}
for all $j=1,\ldots,l$.
We note that in the rigged configuration theory, the order of the riggings associated with
rows of the configuration $\nu$ with the same length is not essential.
Therefore if we have the subset
$\{(\nu_k,I_{k,1}),\ldots,(\nu_k,I_{k,m})\}$ within $(\nu,I)$,
we assume that $I_{k,1}\leq\cdots\leq I_{k,m}$ to erase ambiguity.

In our earlier paper \cite{KiSa14} we conjectured that there is a one to one
correspondence between the set of physical solutions to $\operatorname{BAE}(\ell)$
and the set of the rigged configurations $(\nu,I)$ where the Young diagrams
representing the partitions $\nu$ have $\ell$ cells.
Let us define the flip operator on the set of physical solutions to $\operatorname{BAE}(\ell)$
as follows
\begin{align}
\label{eq:flip}
\kappa:\{\lambda_1,\lambda_2,\ldots,\lambda_\ell\}
\longmapsto
\{-\lambda_1,-\lambda_2,\ldots,-\lambda_\ell\}.
\end{align}
Then another conjecture in \cite{KiSa14} states that the corresponding rigged configuration
are connected by the following operation on the rigged configuration $(\nu,I)$.
Namely, we replace each element $(\nu_j,I_j)$ by $(\nu_j,P_{\nu_j}(\nu)-I_j)$
and reorder the riggings if necessary.
This property is useful in the following discussion.

\section{Bethe's quantum numbers}
\label{sec:bethe}
In what follows, we use function $\Arctan(z)$, defined as an analytic continuation
of the function $\arctan(x)$ ($x\in\mathbb{R}$, $\arctan(0)=0$)
to all complex plane with branch cuts $(i,+i\infty)$ and $(-i,-i\infty)$ along with the imaginary axis.
For example, $\Arctan(xi)=i\operatorname{arctanh}(x)$
and $\Arctan(-xi)=-i\operatorname{arctanh}(x)$ for $x>1$.
If we go across the branch cuts, we have
\begin{align*}
&\Arctan(xi)-\lim_{\epsilon\rightarrow 0}\Arctan(xi-\epsilon)=\pi,\\
&\Arctan(-xi)-\lim_{\epsilon\rightarrow 0}\Arctan(-xi+\epsilon)=-\pi
\end{align*}
for $x>1$ and $\epsilon>0$.

We use the following formula
\begin{align}
\label{eq:arctan}
\log\frac{z-i}{z+i}=2i\Arctan (z)+(2n+1)\pi i
\qquad
(n\in\mathbb{Z}).
\end{align}
We verify the formula as follows.
We see that the differentiations of the both sides give $\frac{2i}{z^2+1}$.
In order to determine the remaining constant,
we compare the values of the both sides at $z=0$ as follows;
\begin{align*}
&\log\frac{0+i}{0-i}=\log(-1)=\log (e^{\pi i}\cdot e^{2n\pi i})=(2n+1)\pi i
\qquad
(n\in\mathbb{Z}),\\
&2i\Arctan(0)=0.
\end{align*}

\subsection{Regular case}
First let us consider a regular solution.
By using the formula (\ref{eq:arctan}),
we take the logarithm of the Bethe ansatz equations (\ref{eq:Bethe_ansatz});
\begin{align*}
&\log\left[
\left(\frac{\lambda_k+\frac{i}{2}}{\lambda_k-\frac{i}{2}}\right)^N
\prod^\ell_{j=1\atop j\neq k}
\frac{\lambda_k-\lambda_j-i}{\lambda_k-\lambda_j+i}
\right]\\
&=-N\cdot 2i\Arctan(2\lambda_k)
+\sum^\ell_{j=1\atop j\neq k}
2i\Arctan (\lambda_k-\lambda_j)
+(-N+\ell-1)\pi i+2n\pi i
\quad
(n\in\mathbb{Z}).
\end{align*}
By the Bethe ansatz equations, this should coincide with
$\log(1)=2m\pi i$ for $m\in\mathbb{Z}$.
Writing $2J_k=2n-2m-N+\ell-1$, we obtain
\begin{align}
\label{eq:bethe}
J_k
=\frac{N}{2\pi}
\left(
2\Arctan(2\lambda_k)
-\frac{2}{N}\sum^\ell_{j =1 \atop j\neq k}\Arctan(\lambda_k-\lambda_j)
\right).
\end{align}
We call the half integers $J_k$ the {\bf Bethe's quantum numbers}.

Here we describe the basic property of the Bethe's quantum numbers in relation
with the multiplication by $(-1)$.
\begin{proposition}
Suppose that we have the following two solutions of the Bethe ansatz equations
\begin{align*}
&\{\lambda_1,\lambda_2,\ldots,\lambda_\ell\},\\
&\{\widetilde{\lambda}_1,\widetilde{\lambda}_2,\ldots,\widetilde{\lambda}_\ell\}
:=\{-\lambda_1,-\lambda_2,\ldots,-\lambda_\ell\}.
\end{align*}
Let $J_\alpha$ (resp. $\widetilde{J}_\alpha$) be the Bethe's quantum number
corresponding to $\lambda_\alpha$ (resp. $\widetilde{\lambda}_\alpha$).
Then we have $J_\alpha=-\widetilde{J}_\alpha$.
\end{proposition}
\begin{proof}
Since $\Arctan(-z)=-\Arctan(z)$, we obtain the result by (\ref{eq:bethe}).
\end{proof}

\subsection{Singular solutions case}
Now let us consider the singular solution (\ref{eq:general_singular_solution})
when $N$ is even.
The computation for $J_k$ ($k\geq 3$) is same with the previous case.
For the computations of $J_1$ and $J_2$, we should be reminded that
the function $\Arctan(z)$ has the logarithmic singularities at $z=\pm i$.
Recall that we assume $\epsilon\in\mathbb{R}$ in (\ref{eq:regularization_of_Nepomechie}).
We introduce the deformation of singular solutions
\begin{align*}
\lambda_1(\theta)
&=\frac{i}{2}+(re^{i\theta})+c(re^{i\theta})^N,\\
\lambda_2(\theta)
&=-\frac{i}{2}+(re^{i\theta})
\end{align*}
where $r\in\mathbb{R}_{>0}$ and study the limits of
\begin{align*}
b_1(\theta)
&=
\frac{N}{2\pi}
\left[
2\Arctan(2\lambda_1(\theta))-\frac{2}{N}\Arctan(\lambda_1(\theta)-\lambda_2(\theta))
\right],\\
b_2(\theta)
&=
\frac{N}{2\pi}
\left[
2\Arctan(2\lambda_2(\theta))-\frac{2}{N}\Arctan(\lambda_2(\theta)-\lambda_1(\theta))
\right]
\end{align*}
when $r\rightarrow 0$.
As a result we obtain the Bethe numbers $J_1(\theta)$ and $J_2(\theta)$
which depend on the choice of $\theta$.
The following two graphs show $J_1(\theta)$ (left) and $J_2(\theta)$ (right)
for the solution $\{i/2,-i/2\}$ of $N=12$.
\begin{center}
\includegraphics[width=70mm]{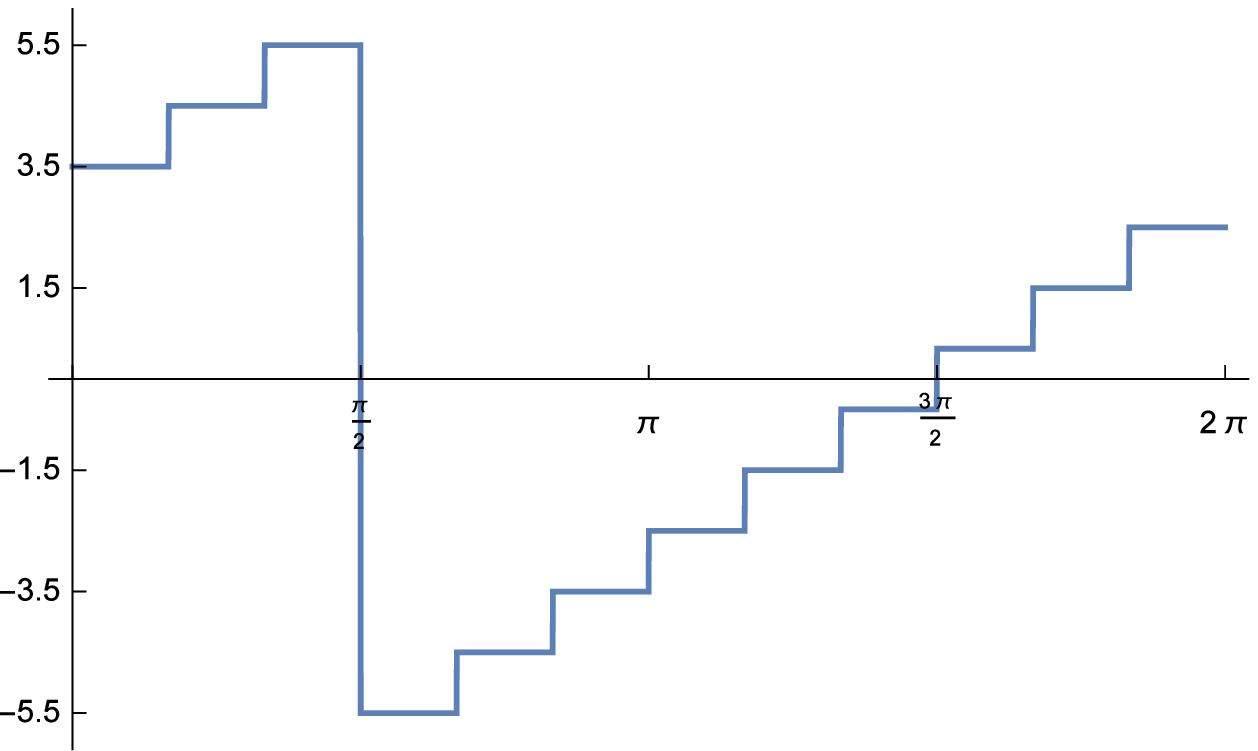}
\quad
\includegraphics[width=70mm]{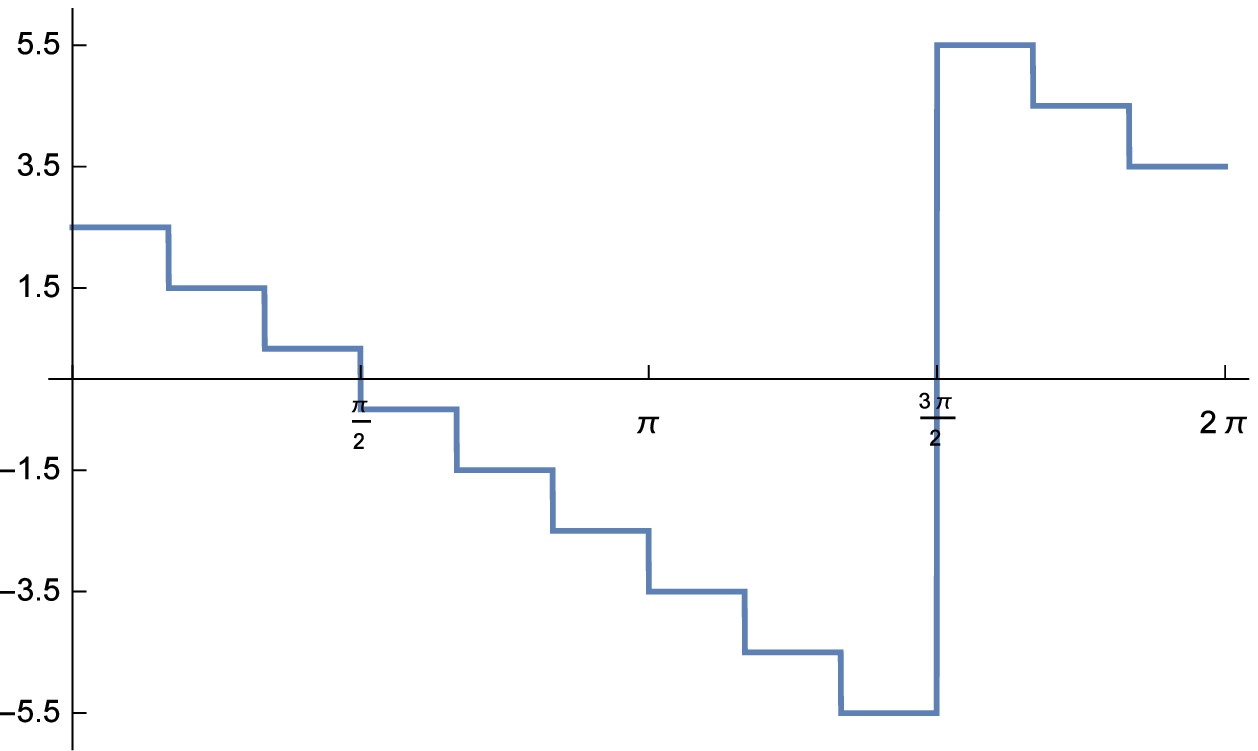}
\end{center}
We note that the sum $J_1(\theta)+J_2(\theta)$ is constant for $-\frac{\pi}{2}<\theta<\frac{\pi}{2}$.

\section{Main algorithm}
\label{sec:main}
We start from a given physical solution
$\lambda=(\lambda_1,\ldots,\lambda_\ell)$
to the Bethe ansatz equations.
We divide $\lambda$ into subsets called strings.
Here we say that $S=\{\eta_1,\ldots,\eta_l\}$, $(\eta_j=a_j+ib_j)$ forms a string if
\[
a_j-a_{j+1}\approx 0,\qquad
b_j-b_{j+1}\approx 1\qquad
(j=1,\ldots,l-1).
\]
By a result of Vladimirov \cite{Vla} one can assume that $\eta_1=\overline{\eta}_l$
if $l>1$.
We define the real part of the string by $\operatorname{Re}(S)=\operatorname{Re}(\eta_1)$.

We take a component of the solution $\lambda$ to the Bethe ansatz equations with maximal
imaginary part and starting from that component we form a string of maximal length.
To proceed, we erase the string created, and apply the above procedure for the remaining components
of the solution $\lambda$.
As a result we partition all component of the solution $\lambda$ into a collection of distinct strings.

One of the main ingredients of our algorithm is the sum of the Bethe numbers corresponding
to a string $S$
\begin{align}
J(S)=\sum_{\eta\in S}J_\eta.
\end{align}
A striking feature of the quantity $J(S)$ is that it exhibits very stable behavior
even if individual Bethe numbers $J_\eta$ behave in a complicated way.

\begin{example}
Let us think about solutions of $N=12$ which contain a 1-string and a 5-string.
We arrange the solutions according to the real parts of the 5-strings.
 We depict the solutions on the complex plane with the Bethe's quantum numbers for each root.
\begin{center}
\unitlength 15pt
\begin{picture}(8.5,9)(-4,-4)
\put(-1.59367,0){\circle*{0.3}}
\put(0.335887,3.18301){\circle*{0.3}}
\put(0.307302,1.5){\circle*{0.3}}
\put(0.307295,0){\circle*{0.3}}
\put(0.307302,-1.5){\circle*{0.3}}
\put(0.335887,-3.18301){\circle*{0.3}}
\put(-4,4){\#1}
\put(0.8,3.2){$7/2$}
\put(0.8,1.7){$11/2$}
\put(0.8,0.2){$5/2$}
\put(0.8,-1.5){$13/2$}
\put(0.8,-3.1){$9/2$}
\put(-4,0.2){$-5/2$}
\put(-4,0){\vector(1,0){8}}
\put(0,-4){\vector(0,1){8}}
\end{picture}
\begin{picture}(8.5,9)(-4,-4)
\put(-0.82796,0){\circle*{0.3}}
\put(0.163391,3.23087){\circle*{0.3}}
\put(0.167061,1.50001){\circle*{0.3}}
\put(0.167056,0){\circle*{0.3}}
\put(0.167061,-1.50001){\circle*{0.3}}
\put(0.163391,-3.23087){\circle*{0.3}}
\put(-4,4){\#2}
\put(0.8,3.2){$13/2$}
\put(0.8,1.7){$7/2$}
\put(0.8,0.2){$1/2$}
\put(0.8,-1.5){$9/2$}
\put(0.8,-3.1){$15/2$}
\put(-3.2,0.2){$-3/2$}
\put(-4,0){\vector(1,0){8}}
\put(0,-4){\vector(0,1){8}}
\end{picture}
\begin{picture}(8.5,9)(-4,-4)
\put(-0.462803,0){\circle*{0.3}}
\put(0.0880767,3.25333){\circle*{0.3}}
\put(0.0955479,1.50001){\circle*{0.3}}
\put(0.0955538,0){\circle*{0.3}}
\put(0.0955479,-1.50001){\circle*{0.3}}
\put(0.0880767,-3.25333){\circle*{0.3}}
\put(-4,4){\#3}
\put(0.8,3.2){$13/2$}
\put(0.8,1.7){$9/2$}
\put(0.8,0.2){$-3/2$}
\put(0.8,-1.5){$11/2$}
\put(0.8,-3.1){$15/2$}
\put(-2.8,0.2){$-1/2$}
\put(-4,0){\vector(1,0){8}}
\put(0,-4){\vector(0,1){8}}
\end{picture}
\end{center}

\begin{center}
\unitlength 15pt
\begin{picture}(8.5,9)(-4,-4)
\put(-0.212095,0){\circle*{0.3}}
\put(0.0396765,3.26352){\circle*{0.3}}
\put(0.0442383,1.50002){\circle*{0.3}}
\put(0.0442652,0){\circle*{0.3}}
\put(0.0442383,-1.50002){\circle*{0.3}}
\put(0.0396765,-3.26352){\circle*{0.3}}
\put(-4,4){\#4}
\put(0.8,3.2){$13/2$}
\put(0.8,1.7){$9/2$}
\put(0.8,0.2){$-3/2$}
\put(0.8,-1.5){$11/2$}
\put(0.8,-3.1){$15/2$}
\put(-2,0.2){$1/2$}
\put(-4,0){\vector(1,0){8}}
\put(0,-4){\vector(0,1){8}}
\end{picture}
\begin{picture}(8.5,9)(-4,-4)
\put(-0.00281903,0){\circle*{0.3}}
\put(0.,3.26659){\circle*{0.3}}
\put(0.,1.50012){\circle*{0.3}}
\put(0.00281903,0){\circle*{0.3}}
\put(0.,-1.50012){\circle*{0.3}}
\put(0.,-3.26659){\circle*{0.3}}
\put(-4,4){\#5}
\put(0.8,3.2){$9/2$}
\put(0.8,1.7){$11/2$}
\put(0.8,0.2){$-3/2$}
\put(0.8,-1.5){$-11/2$}
\put(0.8,-3.1){$-9/2$}
\put(-2,0.2){$3/2$}
\put(-4,0){\vector(1,0){8}}
\put(0,-4){\vector(0,1){8}}
\end{picture}
\begin{picture}(8.5,9)(-4,-4)
\put(0.212095,0){\circle*{0.3}}
\put(-0.0396765,3.26352){\circle*{0.3}}
\put(-0.0442383,1.50002){\circle*{0.3}}
\put(-0.0442652,0){\circle*{0.3}}
\put(-0.0442383,-1.50002){\circle*{0.3}}
\put(-0.0396765,-3.26352){\circle*{0.3}}
\put(-4,4){\#6}
\put(-3,3.2){$-15/2$}
\put(-3,1.7){$-11/2$}
\put(-2.4,0.2){$3/2$}
\put(-3,-1.5){$-9/2$}
\put(-3,-3.1){$-13/2$}
\put(0.8,0.2){$-1/2$}
\put(-4,0){\vector(1,0){8}}
\put(0,-4){\vector(0,1){8}}
\end{picture}
\end{center}

\begin{center}
\unitlength 15pt
\begin{picture}(8.5,9)(-4,-4)
\put(0.462803,0){\circle*{0.3}}
\put(-0.0880767,3.25333){\circle*{0.3}}
\put(-0.0955479,1.50001){\circle*{0.3}}
\put(-0.0955538,0){\circle*{0.3}}
\put(-0.0955479,-1.50001){\circle*{0.3}}
\put(-0.0880767,-3.25333){\circle*{0.3}}
\put(-4,4){\#7}
\put(-3,3.2){$-15/2$}
\put(-3,1.7){$-11/2$}
\put(-2.4,0.2){$3/2$}
\put(-3,-1.5){$-9/2$}
\put(-3,-3.1){$-13/2$}
\put(1,0.2){$1/2$}
\put(-4,0){\vector(1,0){8}}
\put(0,-4){\vector(0,1){8}}
\end{picture}
\begin{picture}(8.5,9)(-4,-4)
\put(0.82796,0){\circle*{0.3}}
\put(-0.163391,3.23087){\circle*{0.3}}
\put(-0.167061,1.50001){\circle*{0.3}}
\put(-0.167056,0){\circle*{0.3}}
\put(-0.167061,-1.50001){\circle*{0.3}}
\put(-0.163391,-3.23087){\circle*{0.3}}
\put(-4,4){\#8}
\put(-3,3.2){$-15/2$}
\put(-3,1.7){$-9/2$}
\put(-3,0.2){$-1/2$}
\put(-3,-1.5){$-7/2$}
\put(-3,-3.1){$-13/2$}
\put(1.5,0.2){$3/2$}
\put(-4,0){\vector(1,0){8}}
\put(0,-4){\vector(0,1){8}}
\end{picture}
\begin{picture}(8.5,9)(-4,-4)
\put(1.59367,0){\circle*{0.3}}
\put(-0.335887,3.18301){\circle*{0.3}}
\put(-0.307302,1.5){\circle*{0.3}}
\put(-0.307295,0){\circle*{0.3}}
\put(-0.307302,-1.5){\circle*{0.3}}
\put(-0.335887,-3.18301){\circle*{0.3}}
\put(-4,4){\#9}
\put(-3,3.2){$-9/2$}
\put(-3,1.7){$-13/2$}
\put(-3,0.2){$-5/2$}
\put(-3,-1.5){$-11/2$}
\put(-3,-3.1){$-7/2$}
\put(2.2,0.2){$5/2$}
\put(-4,0){\vector(1,0){8}}
\put(0,-4){\vector(0,1){8}}
\end{picture}
\end{center}
Let $S$ be the 5-string of each solution.
Then for solutions \#1,\ldots,\#4, we have $\operatorname{Re}(S)>0$ and $J(S)=45/2$.
For solutions \#6,\ldots,\#9, we have $\operatorname{Re}(S)<0$ and $J(S)=-45/2$.
Note that individual Bethe's quantum numbers behave in a rather complicated way.
\qed
\end{example}

In solution \#5 of the above example, we have the roots $\pm a$ ($a\in\mathbb{R}$, $a\ll 1$) near the origin.
Although we cannot decide which root should belong to the 5-string,
this ambiguity does not affect the following procedure.
See solutions \#6 and \#16 of Section \ref{sec:examples} for another type of such ambiguity
which does not affect the procedure also.

\subsection{Configuration}
Starting from the solution $\lambda$, we divide it into strings $S_1,\ldots,S_p$.
Let the lengths of $S_k$ be $l_k$.
Then the partition $\nu$ which appears in the rigged configuration $(\nu,I)$
is $\nu=(l_1,\ldots,l_p)^+$.
Here the symbol $(l_1,\ldots,l_p)^+$ denotes the decreasing order of numbers $(l_1,\ldots,l_p)$.

\subsection{Riggings}
We determine the rigging corresponding to a string in the following way.
If a string $S$ satisfies $\operatorname{Re}(S)=0$,
then the corresponding rigging is the half of the corresponding vacancy number.
Otherwise consider the set of all solutions with a configuration $\nu$, denoted by
\[
\operatorname{Sol}(\nu)=\{\lambda^{(1)},\lambda^{(2)},\ldots\}
=\{\lambda^{(\alpha)}\}.
\]
To start with we compute all Bethe numbers for $\operatorname{Sol}(\nu)$.

The main tool to construct bijection between the set of solutions to the Bethe ansatz equations
and rigged configurations is the following string crossing rule.
Let us consider a string $S$ satisfying $\operatorname{Re}(S)>0$ and length $n$.
Let $S_{i_1},\ldots,S_{i_q}$ be strings such that
\begin{itemize}
\item
length of $S_{i_k}=l_k<n$,
\item
$\operatorname{Re}(S)<\operatorname{Re}(S_{i_k})$
\end{itemize}
for $k=1,\ldots,q$.
Similarly let $S_{j_1},\ldots,S_{j_{q'}}$ be strings such that
\begin{itemize}
\item
length of $S_{j_k}=l_k'>n$,
\item
$\operatorname{Re}(S_{j_k})<\operatorname{Re}(S)$
\end{itemize}
for $k=1,\ldots,q'$.
Then we consider
\begin{align}
\widetilde{J}(S):=
J(S)-\sum_{k=1}^{q}\Delta(l_k,n)+\sum_{k=1}^{q'}\Delta(l_k',n),
\qquad
\Delta(l,n):=l+n-3.
\end{align}

Let us fix some $\nu_i$ of the partition $\nu$ and denote the multiplicity of $\nu_i$ by $m$.
Let $S^{(\alpha)}_{i_{\alpha,k}}$ be the strings of length $\nu_i$ within $\lambda^{(\alpha)}$.
Here we assume that the strings of the same length $\nu_i$ satisfy
$\operatorname{Re}(S^{(\alpha)}_{i_{\alpha,1}})>\operatorname{Re}(S^{(\alpha)}_{i_{\alpha,2}})
>\cdots >\operatorname{Re}(S^{(\alpha)}_{i_{\alpha,m}})$.
Define
\[
\mathcal{M}_{\nu_i,k}=\max_\alpha\{\widetilde{J}(S^{(\alpha)}_{i_\alpha,k})\}.
\]
Suppose that we have $\operatorname{Re}(S^{(\alpha)}_{i_\alpha,k})>0$.
Then the corresponding rigging is
\[
P_{\nu_i}(\nu)-\{\mathcal{M}_{\nu_i,k}-\widetilde{J}(S^{(\alpha)}_{i_\alpha,k})\}.
\]
If we have $\operatorname{Re}(S^{(\alpha)}_{i_\alpha,k})<0$,
we can use the flip operation to obtain the rigging.

Then our main conjecture is the following.
\begin{conjecture}
\label{conj:main}
The above procedure gives one-to-one correspondence between the set of
physical solutions to the Bethe ansatz equations and the rigged configurations.
\end{conjecture}
We confirm the conjecture for the length 12 chain based on the numerical data by \cite{GD2}.
As supplementary data, we provide plots of the solutions to the Bethe ansatz equations
for $N=12$ case (except for simpler ones).
The solutions are arranged according to the real part of the largest string.
For each solution, we assign the numbers corresponding to the solution numbers in \cite{GD2}.

\begin{remark}
Roughly speaking, the riggings represent positions of strings.
Namely the larger rigging corresponds to the larger real part $\operatorname{Re}(S)$.
In \cite{Sakamoto2015}, we observed that in several examples
we can readily read off the riggings if we arrange the solutions according to
the real parts of the largest strings.
We confirm that this algorithm works for all solutions of length 12
if the number of the strings is at most 3.
In other cases, the appearance of the solutions in the ordered set of solutions
becomes more complicated.
Nevertheless we observe there is a strong relation between the riggings
and the positions of the strings even in such case.
\qed
\end{remark}

\subsection{Explicit formulas for $\mathcal{M}$}
We expect that the numbers $\mathcal{M}$ depend only on the length $N$ of the chain
and the shape of the configuration $\nu$.
We make the following conjecture.
If $\nu_2<3$, we have
\begin{align*}
\mathcal{M}_{\nu_1,1}
&=\frac{(N-\nu'_1-\max\{1,\nu'_2\})\nu_1}{2}.
\end{align*}
If $\nu_1>1$ and $\nu_2=1$, we have
\begin{align*}
\mathcal{M}_{1,1}
&=\frac{N+\nu_1+\nu_2-2\nu_1'-3}{2}.
\end{align*}
We checked these formulas for length 12 chain.

\section{Examples}
\label{sec:examples}
We consider the case of $\nu=(3,2,1)$ and $N=12$.
Below we give a list of the solutions depicted on the complex plane
which we arrange according to the real part of the 3-string.
To each root we provide the corresponding Bethe numbers.
Following the diagrams we give tables of the explicit numerical values
which we need in the computation of $\widetilde{J}$.
Note that the solution \#11 contains a singular string.
See also \cite[Section 5]{Sakamoto2015} for an alternative method to determine the riggings.

\paragraph{Group 1.}
\begin{center}
\unitlength 15pt
\begin{picture}(7,7)
\put(0.0,6.2){$\#1$}
\put(4.4,5.3){$7/2$}
\put(4.4,3.3){$5/2$}
\put(4.4,0.2){$9/2$}
\put(0.7,4.4){$-7/2$}
\put(0.7,1.1){$-5/2$}
\put(0.0,2.2){$-5/2$}
\put(4.084,3.){\circle*{0.3}}
\put(4.088,4.992){\circle*{0.3}}
\put(2.084,4.){\circle*{0.3}}
\put(1.57,3.){\circle*{0.3}}
\put(2.084,2.){\circle*{0.3}}
\put(4.088,1.008){\circle*{0.3}}
\put(0,3){\vector(1,0){6}}
\put(3,0){\vector(0,1){6}}
\end{picture}
\begin{picture}(7,7)
\put(0.0,6.2){$\#2$}
\put(4.4,5.3){$7/2$}
\put(4.4,3.3){$5/2$}
\put(4.4,0.2){$9/2$}
\put(0.7,4.4){$-7/2$}
\put(0.7,1.1){$-5/2$}
\put(1.0,2.2){$-3/2$}
\put(4.054,3.){\circle*{0.3}}
\put(4.058,4.992){\circle*{0.3}}
\put(1.718,4.006){\circle*{0.3}}
\put(2.394,3.){\circle*{0.3}}
\put(1.718,1.994){\circle*{0.3}}
\put(4.058,1.008){\circle*{0.3}}
\put(0,3){\vector(1,0){6}}
\put(3,0){\vector(0,1){6}}
\end{picture}
\begin{picture}(7,7)
\put(0.0,6.2){$\#3$}
\put(4.4,5.3){$7/2$}
\put(4.4,3.3){$5/2$}
\put(4.4,0.2){$9/2$}
\put(0.7,4.4){$-7/2$}
\put(0.7,1.1){$-5/2$}
\put(1.0,2.2){$-1/2$}
\put(3.996,3.){\circle*{0.3}}
\put(4.004,4.994){\circle*{0.3}}
\put(2.766,3.){\circle*{0.3}}
\put(1.616,4.01){\circle*{0.3}}
\put(1.616,1.99){\circle*{0.3}}
\put(4.004,1.006){\circle*{0.3}}
\put(0,3){\vector(1,0){6}}
\put(3,0){\vector(0,1){6}}
\end{picture}
\begin{picture}(7,7)
\put(0.0,6.2){$\#4$}
\put(4.3,5.3){$7/2$}
\put(4.3,3.3){$5/2$}
\put(4.3,0.2){$9/2$}
\put(0.7,4.4){$-7/2$}
\put(0.7,1.1){$-5/2$}
\put(3.2,2.2){$1/2$}
\put(3.93,3.){\circle*{0.3}}
\put(3.938,4.996){\circle*{0.3}}
\put(1.562,4.012){\circle*{0.3}}
\put(1.562,1.988){\circle*{0.3}}
\put(3.07,3.){\circle*{0.3}}
\put(3.938,1.004){\circle*{0.3}}
\put(0,3){\vector(1,0){6}}
\put(3,0){\vector(0,1){6}}
\end{picture}
\end{center}

\begin{center}
\unitlength 15pt
\begin{picture}(7,7)
\put(0.0,6.2){$\#5$}
\put(4.2,5.3){$7/2$}
\put(4.2,3.3){$5/2$}
\put(4.2,0.2){$9/2$}
\put(0.6,4.4){$-7/2$}
\put(0.6,1.1){$-5/2$}
\put(3.2,2.2){$3/2$}
\put(3.848,3.){\circle*{0.3}}
\put(3.858,4.998){\circle*{0.3}}
\put(3.386,3.){\circle*{0.3}}
\put(1.524,4.012){\circle*{0.3}}
\put(1.524,1.988){\circle*{0.3}}
\put(3.858,1.002){\circle*{0.3}}
\put(0,3){\vector(1,0){6}}
\put(3,0){\vector(0,1){6}}
\end{picture}
\begin{picture}(7,7)
\put(0.0,6.2){$\#6$}
\put(4.1,5.3){$9/2$}
\put(4.1,3.3){$3/2$}
\put(4.1,0.2){$11/2$}
\put(0.5,4.4){$-7/2$}
\put(0.5,1.1){$-5/2$}
\put(4.1,2.2){$3/2$}
\put(3.768,3.038){\circle*{0.3}}
\put(3.734,4.982){\circle*{0.3}}
\put(1.496,4.014){\circle*{0.3}}
\put(1.496,1.986){\circle*{0.3}}
\put(3.768,2.962){\circle*{0.3}}
\put(3.734,1.018){\circle*{0.3}}
\put(0,3){\vector(1,0){6}}
\put(3,0){\vector(0,1){6}}
\end{picture}
\begin{picture}(7,7)
\put(0.0,6.2){$\#7$}
\put(3.8,5.3){$9/2$}
\put(3.6,3.3){$3/2$}
\put(3.8,0.2){$11/2$}
\put(0.5,4.4){$-7/2$}
\put(0.5,1.1){$-5/2$}
\put(4.6,2.2){$5/2$}
\put(4.656,3.){\circle*{0.3}}
\put(3.46,3.){\circle*{0.3}}
\put(3.46,4.998){\circle*{0.3}}
\put(1.48,4.014){\circle*{0.3}}
\put(1.48,1.986){\circle*{0.3}}
\put(3.46,1.002){\circle*{0.3}}
\put(0,3){\vector(1,0){6}}
\put(3,0){\vector(0,1){6}}
\end{picture}
\end{center}

\[
\begin{array}{l|lll}
\hline\hline
\#1&0.54241927&0.54455699+0.99639165 i&0.54455699-0.99639165 i\\
&-0.45810568+0.50017785 i&-0.45810568-0.50017785 i&-0.71532188\\
\hline
\#2&0.52708058&0.52957875+0.99660493 i&0.52957875-0.99660493 i\\
&-0.64127830+0.50335013 i&-0.64127830-0.50335013 i&-0.30368149\\
\hline
\#3&0.49893578&0.50200196+0.99724969 i&0.50200196-0.99724969 i\\
&-0.69284388+0.50515234 i&-0.69284388-0.50515234 i&-0.11725196\\
\hline
\#4&0.46564665&0.46941736+0.99809739 i&0.46941736-0.99809739 i\\
&-0.71976299+0.50614240 i&-0.71976299-0.50614240 i&0.035044606\\
\hline
\#5&0.42430010&0.42960641+0.99941330 i&0.42960641-0.99941330 i\\
&-0.73869344+0.50683156 i&-0.73869344-0.50683156 i&0.19387395\\
\hline
\#6&0.38490522+0.01906127 i&0.36730804+0.99179719 i&0.36730804-0.99179719 i\\
&-0.75221326+0.50729383 i&-0.75221326-0.50729383 i&0.38490522-0.01906127 i\\
\hline
\#7&0.23056669&0.23083274+0.99967059 i&0.23083274-0.99967059 i\\
&-0.76056174+0.50745313 i&-0.76056174-0.50745313 i&0.82889133\\
\hline
\end{array}
\]

\paragraph{Group 2.}
\begin{center}
\unitlength 15pt
\begin{picture}(7,7)
\put(0.0,6.2){$\#8$}
\put(3.8,5.2){$9/2$}
\put(3.8,3.2){$1/2$}
\put(3.8,0.3){$11/2$}
\put(3.8,4.2){$7/2$}
\put(3.8,1.3){$7/2$}
\put(0.0,2.2){$-5/2$}
\put(3.412,3.){\circle*{0.3}}
\put(3.21,3.998){\circle*{0.3}}
\put(3.41,5.){\circle*{0.3}}
\put(1.34,3.){\circle*{0.3}}
\put(3.21,2.002){\circle*{0.3}}
\put(3.41,1.){\circle*{0.3}}
\put(0,3){\vector(1,0){6}}
\put(3,0){\vector(0,1){6}}
\end{picture}
\begin{picture}(7,7)
\put(0.0,6.2){$\#9$}
\put(3.6,5.2){$11/2$}
\put(3.6,3.2){$1/2$}
\put(3.6,0.3){$13/2$}
\put(3.6,4.2){$5/2$}
\put(3.6,1.3){$5/2$}
\put(0.8,2.2){$-3/2$}
\put(3.216,3.998){\circle*{0.3}}
\put(3.12,4.998){\circle*{0.3}}
\put(2.208,3.){\circle*{0.3}}
\put(3.118,3.){\circle*{0.3}}
\put(3.12,1.002){\circle*{0.3}}
\put(3.216,2.002){\circle*{0.3}}
\put(0,3){\vector(1,0){6}}
\put(3,0){\vector(0,1){6}}
\end{picture}
\begin{picture}(7,7)
\put(0.0,6.2){$\#10$}
\put(3.5,5.2){$11/2$}
\put(3.5,3.2){$1/2$}
\put(3.5,0.3){$13/2$}
\put(3.5,4.2){$5/2$}
\put(3.5,1.3){$5/2$}
\put(0.8,2.2){$-1/2$}
\put(3.138,3.998){\circle*{0.3}}
\put(3.024,4.998){\circle*{0.3}}
\put(2.652,3.){\circle*{0.3}}
\put(3.02,3.){\circle*{0.3}}
\put(3.024,1.002){\circle*{0.3}}
\put(3.138,2.002){\circle*{0.3}}
\put(0,3){\vector(1,0){6}}
\put(3,0){\vector(0,1){6}}
\end{picture}
\begin{picture}(7,7)
\put(0.0,6.2){$\#11$}
\put(3.4,5.2){$9/2$}
\put(3.4,3.2){$-1/2$}
\put(3.4,0.3){$-9/2$}
\put(03.4,2.2){$1/2$}
\put(3.,2.){\circle*{0.3}}
\put(3.,4.){\circle*{0.3}}
\put(3.,3.036){\circle*{0.3}}
\put(3.,4.986){\circle*{0.3}}
\put(3.,2.964){\circle*{0.3}}
\put(3.,1.014){\circle*{0.3}}
\put(0,3){\vector(1,0){6}}
\put(3,0){\vector(0,1){6}}
\end{picture}
\end{center}

\begin{center}
\unitlength 15pt
\begin{picture}(7,7)
\put(0.0,6.2){$\#12$}
\put(0.3,5.2){$-13/2$}
\put(0.3,3.2){$-1/2$}
\put(0.3,0.3){$-11/2$}
\put(0.3,4.2){$-5/2$}
\put(0.3,1.3){$-5/2$}
\put(3.3,2.2){$1/2$}
\put(2.98,3.){\circle*{0.3}}
\put(2.976,4.998){\circle*{0.3}}
\put(2.862,3.998){\circle*{0.3}}
\put(2.862,2.002){\circle*{0.3}}
\put(2.976,1.002){\circle*{0.3}}
\put(3.348,3.){\circle*{0.3}}
\put(0,3){\vector(1,0){6}}
\put(3,0){\vector(0,1){6}}
\end{picture}
\begin{picture}(7,7)
\put(0.0,6.2){$\#13$}
\put(0.3,5.2){$-13/2$}
\put(0.3,3.2){$-1/2$}
\put(0.3,0.3){$-11/2$}
\put(0.3,4.2){$-5/2$}
\put(0.3,1.3){$-5/2$}
\put(3.8,2.2){$3/2$}
\put(2.882,3.){\circle*{0.3}}
\put(2.88,4.998){\circle*{0.3}}
\put(2.784,3.998){\circle*{0.3}}
\put(2.784,2.002){\circle*{0.3}}
\put(2.88,1.002){\circle*{0.3}}
\put(3.792,3.){\circle*{0.3}}
\put(0,3){\vector(1,0){6}}
\put(3,0){\vector(0,1){6}}
\end{picture}
\begin{picture}(7,7)
\put(0.0,6.2){$\#14$}
\put(0.0,5.2){$-11/2$}
\put(0.0,3.2){$-1/2$}
\put(0.0,0.3){$-9/2$}
\put(0.0,4.2){$-7/2$}
\put(0.0,1.3){$-7/2$}
\put(4.7,2.2){$5/2$}
\put(4.66,3.){\circle*{0.3}}
\put(2.79,3.998){\circle*{0.3}}
\put(2.59,5.){\circle*{0.3}}
\put(2.588,3.){\circle*{0.3}}
\put(2.79,2.002){\circle*{0.3}}
\put(2.59,1.){\circle*{0.3}}
\put(0,3){\vector(1,0){6}}
\put(3,0){\vector(0,1){6}}
\end{picture}
\end{center}

\[
\begin{array}{l|lll}
\hline\hline
\#8&0.20669577&0.20597572+1.00038608 i&0.20597572-1.00038608 i\\
&0.10578435+0.50000000 i&0.10578435-0.50000000 i&-0.83021590\\
\hline
\#9&0.059726272&0.06007063+0.99927337 i&0.06007063-0.99927337 i\\
&0.10847310+0.50000000 i&0.10847310-0.50000000 i&-0.39681373\\
\hline
\#10&0.010757119&0.01249979+0.99958901 i&0.01249979-0.99958901 i\\
&0.06941354+0.50000000 i&0.06941354-0.50000000 i&-0.17458378\\
\hline
\#11&0.018539900 i&0.99377501 i&-0.99377501 i\\
&0.50000000 i&-0.50000000 i&-0.018539900 i\\
\hline
\#12&-0.010757119&-0.01249979+0.99958901 i&-0.01249979-0.99958901 i\\
&-0.06941354+0.50000000 i&-0.06941354-0.50000000 i&0.17458378\\
\hline
\#13&-0.059726272&-0.06007063+0.99927337 i&-0.06007063-0.99927337 i\\
&-0.10847310+0.50000000 i&-0.10847310-0.50000000 i&0.39681373\\
\hline
\#14&-0.20669577&-0.20597572+1.00038608 i&-0.20597572-1.00038608 i\\
&-0.10578435+0.50000000 i&-0.10578435-0.50000000 i&0.83021590\\
\hline
\end{array}
\]

\paragraph{Group 3.}
\begin{center}
\unitlength 15pt
\begin{picture}(7,7)
\put(0.0,6.2){$\#15$}
\put(0.0,5.3){$-11/2$}
\put(0.8,3.3){$-3/2$}
\put(0.4,0.2){$-9/2$}
\put(4.5,4.4){$5/2$}
\put(4.5,1.1){$7/2$}
\put(0.0,2.2){$-5/2$}
\put(4.52,1.986){\circle*{0.3}}
\put(4.52,4.014){\circle*{0.3}}
\put(2.54,4.998){\circle*{0.3}}
\put(1.344,3.){\circle*{0.3}}
\put(2.54,3.){\circle*{0.3}}
\put(2.54,1.002){\circle*{0.3}}
\put(0,3){\vector(1,0){6}}
\put(3,0){\vector(0,1){6}}
\end{picture}
\begin{picture}(7,7)
\put(0.0,6.2){$\#16$}
\put(0.0,5.3){$-11/2$}
\put(0.4,3.3){$-3/2$}
\put(0.4,0.2){$-9/2$}
\put(4.5,4.4){$5/2$}
\put(4.5,1.1){$7/2$}
\put(0.4,2.2){$-3/2$}
\put(4.504,4.014){\circle*{0.3}}
\put(2.232,3.038){\circle*{0.3}}
\put(2.266,4.982){\circle*{0.3}}
\put(2.232,2.962){\circle*{0.3}}
\put(2.266,1.018){\circle*{0.3}}
\put(4.504,1.986){\circle*{0.3}}
\put(0,3){\vector(1,0){6}}
\put(3,0){\vector(0,1){6}}
\end{picture}
\begin{picture}(7,7)
\put(0.0,6.2){$\#17$}
\put(0.0,5.3){$-9/2$}
\put(0.0,3.3){$-5/2$}
\put(0.0,0.2){$-7/2$}
\put(4.4,4.4){$5/2$}
\put(4.4,1.1){$7/2$}
\put(1.0,2.2){$-3/2$}
\put(4.476,4.012){\circle*{0.3}}
\put(2.142,4.998){\circle*{0.3}}
\put(2.152,3.){\circle*{0.3}}
\put(2.142,1.002){\circle*{0.3}}
\put(2.614,3.){\circle*{0.3}}
\put(4.476,1.988){\circle*{0.3}}
\put(0,3){\vector(1,0){6}}
\put(3,0){\vector(0,1){6}}
\end{picture}
\begin{picture}(7,7)
\put(0.0,6.2){$\#18$}
\put(0.0,5.3){$-9/2$}
\put(0.0,3.3){$-5/2$}
\put(0.0,0.2){$-7/2$}
\put(4.4,4.4){$5/2$}
\put(4.4,1.1){$7/2$}
\put(1.0,2.2){$-1/2$}
\put(4.438,4.012){\circle*{0.3}}
\put(2.93,3.){\circle*{0.3}}
\put(2.062,4.996){\circle*{0.3}}
\put(2.07,3.){\circle*{0.3}}
\put(2.062,1.004){\circle*{0.3}}
\put(4.438,1.988){\circle*{0.3}}
\put(0,3){\vector(1,0){6}}
\put(3,0){\vector(0,1){6}}
\end{picture}
\end{center}

\begin{center}
\unitlength 15pt
\begin{picture}(7,7)
\put(0.0,6.2){$\#19$}
\put(0.0,5.3){$-9/2$}
\put(0.0,3.3){$-5/2$}
\put(0.0,0.2){$-7/2$}
\put(4.4,4.4){$5/2$}
\put(4.4,1.1){$7/2$}
\put(3.2,2.2){$1/2$}
\put(4.384,4.01){\circle*{0.3}}
\put(1.996,4.994){\circle*{0.3}}
\put(2.004,3.){\circle*{0.3}}
\put(1.996,1.006){\circle*{0.3}}
\put(3.234,3.){\circle*{0.3}}
\put(4.384,1.99){\circle*{0.3}}
\put(0,3){\vector(1,0){6}}
\put(3,0){\vector(0,1){6}}
\end{picture}
\begin{picture}(7,7)
\put(0.0,6.2){$\#20$}
\put(0.0,5.3){$-9/2$}
\put(0.0,3.3){$-5/2$}
\put(0.0,0.2){$-7/2$}
\put(4.4,4.4){$5/2$}
\put(4.4,1.1){$7/2$}
\put(3.5,2.2){$3/2$}
\put(3.606,3.){\circle*{0.3}}
\put(4.282,4.006){\circle*{0.3}}
\put(1.942,4.992){\circle*{0.3}}
\put(1.946,3.){\circle*{0.3}}
\put(1.942,1.008){\circle*{0.3}}
\put(4.282,1.994){\circle*{0.3}}
\put(0,3){\vector(1,0){6}}
\put(3,0){\vector(0,1){6}}
\end{picture}
\begin{picture}(7,7)
\put(0.0,6.2){$\#21$}
\put(0.0,5.3){$-9/2$}
\put(0.0,3.3){$-5/2$}
\put(0.0,0.2){$-7/2$}
\put(3.9,4.4){$5/2$}
\put(3.9,1.1){$7/2$}
\put(4.5,2.2){$5/2$}
\put(4.43,3.){\circle*{0.3}}
\put(3.916,4.){\circle*{0.3}}
\put(1.912,4.992){\circle*{0.3}}
\put(1.916,3.){\circle*{0.3}}
\put(1.912,1.008){\circle*{0.3}}
\put(3.916,2.){\circle*{0.3}}
\put(0,3){\vector(1,0){6}}
\put(3,0){\vector(0,1){6}}
\end{picture}
\end{center}

\[
\begin{array}{l|lll}
\hline\hline
\#15&-0.23056669&-0.23083274+0.99967059 i&-0.23083274-0.99967059 i\\
&0.76056174+0.50745313 i&0.76056174-0.50745313 i&-0.82889133\\
\hline
\#16&-0.38490522+0.01906127 i&-0.36730804+0.99179719 i&-0.36730804-0.99179719i\\
&0.75221326+0.50729383 i&0.75221326-0.50729383 i&-0.38490522-0.01906127 i\\
\hline
\#17&-0.42430010&-0.42960641+0.99941330 i&-0.42960641-0.99941330 i\\
&0.73869344+0.50683156 i&0.73869344-0.50683156 i&-0.19387395\\
\hline
\#18&-0.46564665&-0.46941736+0.99809739 i&-0.46941736-0.99809739 i\\
&0.71976299+0.50614240 i&0.71976299-0.50614240 i&-0.035044606\\
\hline
\#19&-0.49893578&-0.50200196+0.99724969 i&-0.50200196-0.99724969 i\\
&0.69284388+0.50515234 i&0.69284388-0.50515234 i&0.11725196\\
\hline
\#20&-0.52708058&-0.52957875+0.99660493 i&-0.52957875-0.99660493 i\\
&0.64127830+0.50335013 i&0.64127830-0.50335013 i&0.30368149\\
\hline
\#21&-0.54241927&-0.54455699+0.99639165 i&-0.54455699-0.99639165 i\\
&0.45810568+0.50017785 i&0.45810568-0.50017785 i&0.71532188\\
\hline
\end{array}
\]

Let us explain how our algorithm works for the examples listed.

\paragraph{3-strings.}
For solutions \#1, \ldots, \#5, we have $\widetilde{J}=21/2$
since there are no shorter strings on the right of the 3-string.

In solution \#6, we have two roots $0.38\pm 0.02i$.
We classify one of them to 1-string and the other to 3-string.
In any case, 1-string is on the right of the 3-string.
Thus we have $\widetilde{J}=23/2-\Delta(3,1)=21/2$.
Similarly we have $\widetilde{J}=21/2$ in solution \#7.

In solution \#8, we have $\tilde{J}=21/2$ since there are no shorter
string on the right.

In solutions \#9 and \#10 we have $\widetilde{J}=25/2-\Delta(3,2)=21/2$
since the 2-string is on the right of the 3-string.

In solution \#11, the real part of the 3-string is 0.
Thus the corresponding rigging is the half of the corresponding vacancy number
$P_3(3,2,1)=0$, that is, the rigging is 0.

In solutions \#12, \ldots, \#21, we have $\widetilde{J}<0$.

Summarizing, we have $\mathcal{M}_3=21/2$.
Since we have $\widetilde{J}=\pm 21/2$ (except for the solution \#11),
we conclude that the corresponding rigging is always $0$.

\paragraph{2-strings.}
In solutions \#1, \ldots, \#7, \#12, \#13, and \#14,
the real part of the 2-strings are negative.

In solution \#8, we have $\tilde{J}=14/2$ since the 3-string is on the right
and the 1-string is on the left.

In solutions \#9 and \#10, we have $\widetilde{J}=10/2+\Delta(2,3)=14/2$
since the 3-string is on the left of the 2-string.

In solutions \#15, \ldots, \#21, we have $\widetilde{J}=12/2+\Delta(2,3)=16/2$
since the 3-string is on the left of the 2-string.
Note that the position of the 1-string does not affect $\widetilde{J}$
since we have $\Delta(2,1)=0$.

To summarize, we have $\mathcal{M}_2=16/2$.
Thus the riggings for the two strings are 0 for solutions in Group 1,
1 for solutions in Group 2 and 2 for the solutions in Group 3.

\paragraph{1-strings.}
By a similar argument we see that $\mathcal{M}_1=7/2$.
Then the riggings for each group is $1,2,\ldots,6$ from the smaller solutions number
to the larger solutions number.

The resulting rigged configurations agree with the ones obtained by a geometrical argument
\cite{Sakamoto2015} or by the computation of Takahashi's quantum numbers \cite{GD2}
(see also \cite{DG1}).

\begin{remark}
The above example reminds us about the box-ball system
(see \cite{Sakamoto_review} for a review).
The box-ball system is a discrete soliton system
where solitons are ``crystal" analogue of magnons.
One of the important properties is that the rigged configurations provide
action and angle variables of the box-ball system.
In this picture each row $\nu_i$ of the configuration $\nu$
represents a soliton of the same length.
If we look at the above examples, it is tempting to consider
them as a propagation of the string of solutions where scattering of
strings yields a phase shift $\Delta(m,n)=m+n-2$.
Note that the phase shift in the case of the box-ball system is $\min(m,n)$
for the scattering of lengths $m$ and $n$ solitons.
It will be interesting to remark that in the box-ball system the crystal analogue of
the transfer matrices, which provides the quantum integrability of the box-ball system,
generates the corresponding rigged configurations which are dynamical variables
if we regard the box-ball system as classical integrable system
\cite{Sakamoto2007}.
\qed
\end{remark}

\section{Conclusion}
\label{sec:conclusion}
Summarizing, we describe an algorithm which associate collection of quantum numbers or riggings
to solutions to the Bethe ansatz equations for the spin-1/2 Heisenberg model,
starting from the set of Bethe's quantum numbers corresponding to those solutions
and its string structure for general $N$.
We confirm our algorithm for all solutions to the Bethe ansatz equations for $N=12$.


\begin{thebibliography}{99}
\bibitem[AV]{AV}
L.~V.~Avdeev and A.~A.~Vladimirov,
{\it Exceptional solutions to the Bethe ansatz equations},
Theor. Math. Phys. {\bf 69} (1986) 1071--1079.

\bibitem[BMSZ]{Beisert}
N.~Beisert, J.~A.~Minahan, M.~Staudacher and K.~Zarembo,
{\it Stringing spins and spinning strings},
JHEP {\bf 09} (2003) 010 (27pp).

\bibitem[B]{Bethe}
H.~A.~Bethe,
{\it Zur theorie der metalle},
Zeit. f\"{u}r Physik {\bf 71} (1931) 205--226.

\bibitem[DG14]{DG1} T.~Deguchi and P.~R.~Giri, 
\textit{Non self-conjugate strings, singular strings and rigged configurations in the Heisenberg model},
 J. Stat. Mech: Theor. Exp. (2015) P02004.
 
\bibitem[DG15]{DG2} T.~Deguchi and P.~R.~Giri,
\textit{Exact quantum numbers of collapsed and non-collapsed
2-string solutions in the Heisenberg spin chain},
arXiv:1509.00108

\bibitem[EKS]{EKS:1992}
F.~H.~L.~Essler, V.~E.~Korepin and K.~Schoutens,
{\it Fine structure of the Bethe ansatz for the
spin-$\frac{1}{2}$ Heisenberg $XXX$ model},
J. Phys. A: Math. Gen. {\bf 25} (1992) 4115--4126.

\bibitem[F]{Faddeev}
L.~D.~Faddeev,
{\it How algebraic Bethe ansatz works for integrable model},
arXiv:hep-th/9605187 (Les-Houches lectures).

\bibitem[GD]{GD2}
P.~R.~Giri and T.~Deguchi,
{\it Heisenberg model and rigged configurations},
J. Stat. Mech: Theor. Exp. P07007 (2015).

\bibitem[HC]{HC}
R.~Hagemans and J.-S.~Caux,
{\it Deformed strings in the Heisenberg model},
J. Phys. A: Math. Theor. {\bf 40} (2007) 14605--14647.

\bibitem[HNS]{HNS1}
W.~Hao, R.~I.~Nepomechie and A.~I.~Sommese,
{\it Completeness of solutions of Bethe's equations},
Phys. Rev. E {\bf 88} (2013) 052113 (8pp plus supplemental material).

\bibitem[KR]{KR}
A.~N.~Kirillov and N.~Yu.~Reshetikhin:
{\it The Bethe ansatz and the combinatorics of Young tableaux},
Zap. Nauch. Sem. LOMI {\bf 155} (1986) 65--115.
(English Translation: J. Soviet Math. {\bf 41} (1988) 925--955.)

\bibitem[KS14a]{KiSa14}
A.~N.~Kirillov and R.~Sakamoto,
{\it Singular solutions to the Bethe ansatz equations and rigged configurations},
J. Phys. A: Math. Theor. {\bf 47} (2014) 205207 (20pp).

\bibitem[KS14b]{KirillovSakamoto2014b}
A.~N.~Kirillov and R.~Sakamoto,
{\it Some remarks on Nepomechie--Wang eigenstates for spin 1/2 XXX model},
Moscow Math. J. {\bf 15} (2015) 337--352.

\bibitem[KBI]{KorepinBook}
V.~E.~Korepin, N.~M.~Bogoliubov and A.~G.~Izergin,
{\it Quantum inverse scattering method and correlation functions},
Cambridge University Press (1993).

\bibitem[LR]{LinR}
S.-S. Lin and S.-S. Roan,
{\it Bethe ansatz for Heisenberg XXX model},
arXiv:cond-mat/9509183 

\bibitem[NW]{NepomechieWang2013}
R.~I.~Nepomechie and C.~Wang,
{\it Algebraic Bethe ansatz for singular solutions},
J. Phys. A: Math. Theor. {\bf 46} (2013) 325002 (8pp).

\bibitem[S07]{Sakamoto2007}
R.~Sakamoto,
{\it Kirillov--Schilling--Shimozono bijection as energy functions of crystals},
International Mathematics Research Notices (2009) 2009: 579--614.

\bibitem[S12]{Sakamoto_review}
R.~Sakamoto,
{\it Ultradiscrete soliton systems and combinatorial representation theory},
RIMS Kokyuroku {\bf 1913} (2014) 141--158.
(Also available as arXiv:1212.2774)

\bibitem[S15]{Sakamoto2015}
R.~Sakamoto,
{\it Rigged configurations approach for the spin-1/2 isotropic Heisenberg model},
J. Phys. A: Math. Theor. 48 (2015) 165201 (25pp).

\bibitem[T]{Takahashi}
M.~Takahashi,
{\it One-dimensional Heisenberg model at finite temperature},
Prog. Theor. Phys. {\bf 46} (1971) 401--415.

\bibitem[V]{Vla}
A.~A.~Vladimirov,
{\it Proof of the invariance of the Bethe-ansatz solutions under complex conjugation},
Theor. Math. Phys. {\bf 66} (1986) 102--105.

\end{thebibliography}
\end{document}